\documentclass[11pt, letterpaper]{article}

\usepackage{amssymb,amsmath,amsthm,thmtools,mathtools}
\usepackage{color,xcolor,soul}
\usepackage[colorlinks=true,linkcolor=purple,citecolor=blue]{hyperref}
\usepackage[nameinlink]{cleveref}
\usepackage{booktabs}
\usepackage{makecell}
\usepackage{multirow}
\definecolor{lightgreen}{RGB}{200,255,200}
\sethlcolor{lightgreen}
\usepackage[letterpaper,margin=1in]{geometry}
\newtheorem{theorem}{Theorem}
\newtheorem{corollary}{Corollary}
\newtheorem{lemma}{Lemma}
\newtheorem{proposition}{Proposition}

\makeatletter
\let\c@proposition\c@theorem
\let\c@corollary\c@theorem
\let\c@lemma\c@theorem
\makeatother

\theoremstyle{definition}

\newtheorem*{theorem*}{Theorem}
\usepackage{tikz}
\tikzstyle{vertex}=[circle, draw, inner sep=0pt, minimum size=6pt]

\usetikzlibrary {arrows.meta}
\usepackage{algorithm,algpseudocode}
\usepackage{nicefrac,xspace}

\newcommand{\e}{\varepsilon}

\newcommand{\TT}{\mathbb{T}}
\newcommand{\T}{{\cal T}}

\newcommand{\ba}{{\mathbf a}}
\newcommand{\bq}{{\mathbf q}}
\newcommand{\mfvs}{\mathsf{mfvs}} 
\newcommand{\dfvs}{{\sf DFVS}\xspace}

\title{A deterministic $(2+\e)$-approximation\\ for directed feedback vertex sets in tournaments} 
\date{09.07.2026}

\begin{document}
\author{Ebrahim Ghorbani\thanks{Hamburg University of Technology, Institute for Algorithms and Complexity, Hamburg, Germany. \texttt{ebrahim.ghorbani@tuhh.de}.} \and Matthias Mnich\thanks{Hamburg University of Technology, Institute for Algorithms and Complexity, Hamburg, Germany. \texttt{matthias.mnich@tuhh.de}}}

	\maketitle
  \thispagestyle{empty}

\begin{abstract}
  We nearly settle the polynomial-time approximability of the {\sc Directed Feedback\linebreak Vertex Set} problem in tournaments.
  This problem is {\sc Vertex Cover}-hard, and thus cannot have a $(2-\varepsilon)$-approximation for any $\varepsilon > 0$ in polynomial time assuming the Unique Games Conjecture.
  
  In the past 28 years, several works have attempted to attain this approximability barrier of~2, and have designed algorithms with smaller and smaller approximation factors.
  This includes a $\nicefrac{5}{2}$-approximation by Cai, Deng and Zang (FOCS 1998, SICOMP 2001); a\linebreak $\nicefrac{7}{3}$-approximation by Mnich, Vassilevska Williams and V{\'e}gh (ESA 2016), another $\nicefrac{7}{3}$-approxima\-tion by Aprile, Drescher, Fiorini and Huynh (DAM 2023), and a $\nicefrac{9}{4}$-approximation by Ghorbani and Mnich (ICALP 2026).
  
  Our main result improves upon all of those works: we give the first deterministic polynomial-time $(2+\varepsilon)$-approximation for {\sc Directed Feedback Vertex Set} in tournaments, for all~\mbox{$\varepsilon > 0$}.
  We thereby almost answer an open question by Lokshtanov, Misra, Mukherjee, Panolan, Philip and Saurabh (SODA 2020) who asked for a deterministic 2-approximation in polynomial time.
  
	Furthermore, we extend our result to the broader class of quasi-transitive digraphs.
\end{abstract}

\clearpage
\pagebreak
\setcounter{page}{1}

\section{Introduction}
\label{sec:introduction}
One of the most intensely studied problems in algorithmic graph theory is {\sc Vertex Cover}, which for a given undirected graph $G$ seeks a minimum-sized subset $S\subseteq V(G)$ which intersects all edges of~$G$.
Due to its $\mathsf{NP}$-hardness, one is interested in designing $\alpha$-approximation algorithms for it, which in polynomial time return a vertex cover whose size is at most $\alpha$ times that of a minimum-sized vertex cover.
A 2-approximation algorithm for {\sc Vertex Cover} is known since the 1980s.
Under the Unique Games Conjecture, no $(2-\varepsilon)$-approximation algorithm can exist~\cite{KhotR2008}, for any~$\varepsilon > 0$.

A generalization of {\sc Vertex Cover} on graphs is {\sc Vertex Cover} on hypergraphs.
For 3-uniform hypergraphs, where every hyperedge has size 3, the problem is also known as {\sc 3-Hitting Set}.
{\sc 3-Hitting Set} admits a 3-approximation algorithm, and there cannot exist a $(3-\varepsilon)$-approximation algorithm for any $\varepsilon > 0$, again assuming the truth of the Unique Games Conjecture~\cite{GuruswamiL2014}.
Research has thus focused on identifying relevant classes of {\sc 3-Hitting Set} instances for which $\alpha$-approxi\-mation algorithms with $\alpha < 3$ can be obtained.
V{\'e}gh~\cite{Vegh} explitly posed this as an open question.
A complete answer to this question seems currently not known.

One such well-investigated class is the {\sc Directed Feedback Vertex Set} problem in tournaments, in its unweighted and node-weighted variants.
Since the problem was shown to be at least as hard as {\sc Vertex Cover} in 1990 by Speckenmeyer~\cite{Speckenmeyer1990}, the past 28 years have seen several approximation algorithms for it, see~\Cref{tab:fvstapprox}.
\begin{table}[h!]
  \centering
  \small
  \begin{tabular}{lrlll}
    \toprule
    Reference & \multicolumn{2}{l}{approx. factor} & run time & randomized/det. \\
    \midrule
    Cai, Deng, and Zang~\cite{CaiDZ1998,CaiDZ2001} & ~ & $\nicefrac{5}{2}$ & $n^{O(1)}$ & det.\\
    Mnich, Vassilevska Williams, and V{\'e}gh~\cite{MnichVWV2016} & ~ & $\nicefrac{7}{3}$ & $n^{O(1)}$ & det.\\
    Aprile, Drescher, Fiorini, and Huynh~\cite{AprileDFH2023} & ~ & $\nicefrac{7}{3}$ & $n^{O(1)}$ & det.\\
    Lokshtanov, Misra, Mukherjee,& \multirow{2}{*}{\huge $\{$} & $2$ & $n^{O(1)}$ & rand., w.p. $7/10$\\
    $\qquad$ Panolan, Philip, and Saurabh~\cite{LokshtanovMMPPS2020} & & $2$ & $n^{O(\log n)}$ & det.\\
    Ghorbani and Mnich~\cite{GhorbaniM2026} & ~ & $\nicefrac{9}{4}$ & $n^{O(1)}$ & det.\\
    \emph{This paper} & ~ & $2+\e$ & $n^{O(1)}$ & det.\\
    \bottomrule
  \end{tabular}
  \vspace{-0.7em}
  \caption{Approximation algorithms for {\sc Directed Feedback Vertex Set} in tournaments.\label{tab:fvstapprox}}
\end{table}
When this stream of research was initiated by Cai, Deng and Zang~\cite{CaiDZ2001} in 1998, the major goal has been the design of a polynomial-time (deterministic) 2-approximation algorithm for {\sc Directed Feedback Vertex Set} in tournaments.
This would match the approximation factor of the {\sc Vertex Cover} problem in graphs, which it generalizes, and would give the best constant approximation guarantee under the Unique Games Conjecture.
The quest to obtain such an algorithm was explicitly stated by Mnich, Vassilevska Williams and V{\'e}gh~\cite{MnichVWV2016}.
It was re-iterated in a breakthrough paper by Lokshtanov, Misra, Mukherjee, Panolan, Philip and Saurabh~\cite{LokshtanovMMPPS2020}, who gave a \emph{randomized} 2-approximation for the problem that succeeds with probability $7/10$, and which can be derandomized at the cost of taking \emph{quasi-polynomial} time.

\subsection{Our contributions}
\label{sec:ourcontributions}
Our main result nearly answers the quest for a deterministic 2-approximation algorithm for (node-weighted) {\sc Directed Feedback Vertex Set} in tournaments.
We thereby improve upon the long series of deterministic $\alpha$-approximation algorithms from~\Cref{tab:fvstapprox}, for absolute constants $\alpha > 2$.
\begin{theorem}
\label{thm:2+epsilon-tournament}
  There is a deterministic polynomial-time algorithm for {\sc Directed Feedback \mbox{Vertex} Set} in tournaments, which gives a $(2+\varepsilon)$-approximation for any $\varepsilon > 0$.
\end{theorem}

A digraph is \emph{quasi-transitive} if, 
for every node $v$, there is complete adjacency between the in-neighbors and the out-neighbors of $v$.
Quasi-transitive digraphs are of particular interest because of their close connection with comparability graphs.
Indeed, a graph admits a quasi-transitive orientation if and only if it is a comparability graph~\cite{BangJensenH1995}.
Clearly, every tournament is quasi-transitive.
Quasi-transitive digraphs have been studied extensively; for example, the standard textbook on digraphs by Bang-Jensen and Gutin contains an entire chapter on this class~\cite[Ch.~8]{BangJensenG2018}.
We further show how to extend our deterministic $(2+\e)$-approximation for node-weighted {\sc Directed Feedback Vertex Set} from tournaments to quasi-transitive digraphs.
\begin{theorem}
	\label{cor:dfvs-appx-quasi-transitive}
	For every $\e>0$, the node-weighted {\sc Directed Feedback Vertex Set} problem on quasi-transitive digraphs admits a deterministic $(2+\e)$-approximation in polynomial time.
\end{theorem}

\subsection{Our methodology}
A natural approach to obtain a good approximation algorithm, with approximation factor $\alpha$ close to or equal to 2, is to lower-bound the weight of optimal directed feedback vertex set $F^\star$ by the optimal solution value of the linear programming relaxation of the natural covering integer programming formulation of the problem which stipulates that at least one node of each directed triangle should be included into $F^\star$.
Indeed, some of the previous deterministic algorithms for the problem solve some linear program as a subroutine~\cite{GhorbaniM2026,MnichVWV2016}, or even some Sherali-Adams hierarchies of it~\cite{AprileDFH2023}.
However, a careful analysis of the integrality gap shows that it is at least $\nicefrac{7}{3}$, even when the linear program is augmented with constraints that require $F^\star$ to include 3 nodes of each $\mathcal T_7$-subtournament.
Here, $\mathcal{T}_{2k+1}$ denotes the family of tournaments on $2k+1$ nodes whose minimum directed feedback vertex set has size at least $k$.
A tournament is said to be \emph{$\mathcal{T}_{2k+1}$-free} if it contains no subtournament belonging to $\mathcal{T}_{2k+1}$.
Removing all copies of $\mathcal T_7$-subtournaments (of which there are 121 non-isomorphic ones) by rounding LP solutions, and understanding the structure of $\mathcal T_7$-free tournaments, is crucial to both $\nicefrac{7}{3}$-approximations.
The next result in line is a $\nicefrac{9}{4}$-approximation by Ghorbani and Mnich~\cite{GhorbaniM2026}.
In a first step, it removes copies of $\mathcal T_9$-subtournaments.
In a second step, for $\mathcal T_9$-free tournaments, they then analyzed the structure of the digraph induced by the arcs which belong to triangles.
Instead of a layering algorithm, they employed fundamental results from partial order theory such as the Greene-Kleitman Theorem to obtain their $\nicefrac{9}{4}$-approximation.
The first removal step alone causes a loss of $\nicefrac{9}{4}$ in the approximation factor, and thus cannot lead to an $\alpha$-approximation algorithm for any $\alpha < \nicefrac{9}{4}$.

Interestingly, our novel algorithm does not solve any linear program, and also does not employ the Greene-Kleitman Theorem.
Instead, we provide a purely combinatorial algorithm.
To overcome the impediment of the $\nicefrac{9}{4}$-approximation barrier, a natural step is to instead remove all copies of $6$-node subtournaments whose minimum feedback vertex sets have size at least 3.
There is only one such tournament, which we denote by $T_6$, 
and removing all copies of it does not harm the approximation factor of 2.
However, it appears that $T_6$-free tournaments do not possess enough structure to devise any algorithmic approaches for obtaining a 2-approximation on them.
In particular, $T_6$ is not a ``hero''.
The notion of hero was introduced by Berger, Choromanski, Chudnovsky, Fox, Loebl, Scott, Seymour, and Thomass{\'e}~\cite{BergerCCFLSST2013} when investigating the cycle structure of tournaments.
A tournament $H$ is a \emph{hero} if there exists a constant $c_H$ such that every $H$-free tournament $T$ has \emph{dichromatic number} $\vec{\chi}(T)$ at most $c_H$, that is, the nodes of $T$ can be partitioned into at most~$c_H$ acyclic subtournaments.
Their main result is a characterization of all heroes, a consequence of which is that~$T_6$ is \emph{not} a hero, and thus the dichromatic number of $T_6$-free tournaments cannot be bounded.

As the first part of our approach, we establish a hero-type property. Specifically, we prove that for every $k \ge 1$, every $\T_{2k+1}$-free tournament satisfies a stronger structural property which, in particular, implies that its dichromatic number is bounded by a constant $c_k$ depending only on $k$.

For a tournament $T$, its \emph{triangle digraph} $H(T)$ is the spanning subdigraph of $T$ whose arc set consists of the arcs of $T$ that belong to directed triangles in $T$.
In particular, any node of~$T$ that is not contained in any directed triangle is an isolated node of $H(T)$.
We first use a key property of the triangle digraph that was recently established by Ghorbani and Mnich~\cite{GhorbaniM2026}, namely, that its underlying undirected graph is perfect.
By considering further forbidden-subtournament arguments, we obtain a bound on the chromatic number of the triangle digraph of $\T_{2k+1}$-free tournaments.
In what follows, a \emph{proper $t$-coloring} of $H(T)$ refers to a proper $t$-node-coloring, that is, a partitioning of the node set into $t$ independent sets.

\begin{theorem}
\label{thm:t(k)-bound-chi(H)}
  For every integer $k\ge1$, there exists an integer $t=t_k$ such that for every $\T_{2k+1}$-free tournament $T$ its triangle digraph $H(T)$ admits a proper $t$-coloring.
\end{theorem}	

In particular, \Cref{thm:t(k)-bound-chi(H)} implies the aforementioned hero-type property of the family $\T_{2k+1}$. 

The color classes $C$ of $H(T)$ induce subtournaments in $T$ with a highly restricted structure, implying that any transitive subtournament of $T$ intersects each $C$ in a \emph{prefix} of $C$.
This property gives rise to a dynamic programming approach which computes a maximum-weight subtournament, and thus a minimum-weight feedback vertex set of $T$, in polynomial time for fixed~$t$.

\begin{theorem}
\label{thm:dynamic-program}
  Let $T$ be a tournament of order $n$ with non-negative node weights $w:V(T)\to \mathbb{Q}_{\ge 0}$.
  There	is an algorithm that, given a proper $t$-coloring of $H(T)$, computes a minimum-weight directed feedback vertex set	of $(T,w)$ in time~$O(tn^{t+1})$.
\end{theorem}

The perfectness of $H(T)$ allows us to compute a proper $t$-coloring of $H(T)$ efficiently.
This in turn enables us to solve the node-weighted {\sc Directed Feedback Vertex Set} problem exactly on
$\T_{2k+1}$-free tournaments.
\begin{corollary}
\label{cor:poly-DFVS-Tk-free}
  For every $\T_{2k+1}$-free weighted tournament $(T,w)$ on $n$ nodes, a minimum-weight directed feedback vertex set of $(T,w)$ can be computed in polynomial time. 
\end{corollary}
Note that this corollary alone improves and simplifies a number of previous works.
For instance, Cai, Deng and Zang~\cite{CaiDZ2001} gave a very intricate 10-page analysis to obtain a polynomial-time algorithm for {\sc Directed Feedback Vertex Set} in $\mathcal T_5$-free tournaments.
The special case of $k=2$ of \Cref{cor:poly-DFVS-Tk-free} recovers a simple algorithm for the same problem.
Also, for $\mathcal T_7$-free tournaments, Mnich, Vassilevska Williams and V{\'e}gh~\cite{MnichVWV2016} gave a polynomial-time $\nicefrac{7}{3}$-approximation, based on a complex layering algorithm.
A simplified layering algorithm was given by Aprile, Drescher, Fiorini and Huynh~\cite{AprileDFH2023}, which also yields a $\nicefrac{7}{3}$-approximation in $\mathcal T_7$-free tournaments.
Both were improved to a polynomial-time $2$-approximation by Ghorbani and Mnich~\cite{GhorbaniM2026}.
As a special case with $k=3$, \Cref{cor:poly-DFVS-Tk-free} improves upon all of these results by providing an \emph{exact} polynomial-time algorithm for $\mathcal T_7$-free tournaments (without requiring any layering algorithm).

Finally, given $\varepsilon>0$, we choose an integer $k\ge 1/\varepsilon$ and combine this exact algorithm with a local-ratio preprocessing procedure that repeatedly identifies and removes weighted copies of subtournaments from $\T_{2k+1}$, yielding a deterministic $(2+\e)$-approximation.
This local-ratio procedure resembles similar procedures that were used in the $\nicefrac{5}{2}$-approximation~\cite{CaiDZ2001} 
to all values of $k$.
In fact, the procedure for the $\nicefrac{5}{2}$-approximation was highlighted as an important application of the local-ratio paradigm in the influential survey by Bar-Yehuda, Bendel, Freund and Rawitz~\cite{BarYehudaBFR2004}.

\subsection{Related work}
As mentioned, currently there is no full characterization known for which classes of 3-uniform hypergraphs the {\sc 3-Hitting Set} problem admits a polynomial-time $\alpha$-approximation algorithm for some absolute constant $\alpha < 3$.

On the one hand, there is a related {\sc Vertex Cover}-hard problem which clearly illustrates the challenge to obtain such approximation algorithms: {\sc Cluster Vertex Deletion}.
There, the goal is to hit all 3-node 2-edge induced paths in a given undirected graph with a minimum number of node deletions.
Its history of approximation algorithms almost parallels that for {\sc Directed Feedback Vertex Set} in tournaments: first a $\nicefrac{5}{2}$-approximation~\cite{YouWC2017}, then a $\nicefrac{7}{3}$-approximation~\cite{FioriniJS2016}, then a $\nicefrac{9}{4}$-approximation~\cite{FioriniJS2020}, until eventually a 2-approximation was obtained~\cite{AprileDFH2023b}.

On the other hand, for another related problem the existence of non-trivial approximation algorithms seems unlikely.
Namely, the problem of hitting undirected triangles in undirected graphs by a minimum number of nodes cannot have a $(3-\varepsilon)$-approximation for any $\varepsilon > 0$ assuming the truth of the Unique Games Conjecture~\cite{GuruswamiL2014}.

The structure of directed feedback vertex sets in tournaments has also been well-investigated from three related algorithmic concepts.
The first one of those are exact exponential-time algorithms; the asymptotically fastest one computes a minimum directed feedback vertex set of any $n$-node tournament in time $O(1.466^n)$~\cite{LokshtanovK2016}.
The second one are fixed-parameter algorithms, which aim to decide whether a given $n$-node tournament admits a directed feedback vertex set of size $k$; here the algorithm with the most graceful dependence on $k$ requires time $1.618^k\cdot n^{O(1)}$~\cite{LokshtanovK2016}.
And third, enumeration algorithms which list all inclusion-minimal directed feedback vertex sets of a given $n$-node tournament; this time is known to be bounded by $O(1.595^n)$~\cite{MnichT2018}.
 
\noindent
\textbf{Organization.}
The rest of the paper is organized as follows.
In \Cref{sec:trianglegraph}, we collect the main properties of the triangle digraph of a tournament, following Ghorbani and Mnich~\cite{GhorbaniM2026}, and prove that, for fixed $k$, the triangle digraph of every $\T_{2k+1}$-free tournament has bounded chromatic number, thus establishing \Cref{thm:t(k)-bound-chi(H)}. 
In \Cref{sec:dynamic-programming}, we solve the node-weighted {\sc Directed Feedback Vertex Set} problem on $\T_{2k+1}$-free tournaments via dynamic programming, and prove  \Cref{thm:dynamic-program} and \Cref{cor:poly-DFVS-Tk-free}. 
In \Cref{sec:2+epsilon-appx}, we give a weighted local-ratio algorithm that yields a deterministic $(2+\e)$-approximation for node-weighted {\sc Directed Feedback Vertex Set}, and prove \Cref{thm:2+epsilon-tournament} and \Cref{cor:dfvs-appx-quasi-transitive}.
Finally, in \Cref{sec:runtime-improve}, we show how to improve the run time of our 
 $(2+\varepsilon)$-approximation algorithm.

\section{The triangle digraph of tournaments}
\label{sec:trianglegraph}
As part of their $\nicefrac{9}{4}$-approximation for {\sc Directed Feedback Vertex Set} in quasi-transitive digraphs, Ghorbani and Mnich~\cite{GhorbaniM2026} studied the triangle digraph of a tournament and analyzed its properties in relation to the {\sc Directed Feedback Vertex Set} problem.
In particular, they showed that the triangle digraph of a tournament is \emph{perfect} (that is, in every induced subdigraph, the chromatic number equals the clique number, where for a digraph, the clique number and chromatic number refer to those of its underlying undirected graph, see the digraphs book by Bang-Jensen and Gutin~\cite[Section~11.7]{BangJensenG2018}).
In the following, we build upon their concept and utilize the properties of the triangle digraph.
As the key ingredient of our approach, we prove that considering tournaments not containing a subtournament with a relatively large minimum directed feedback vertex set forces their triangle digraph to have bounded clique number.
	
For a digraph $D$, let $V(D)$ be its node set and let $A(D)$ be its arc set, and write $u \to v$ to indicate that $(u,v)\in A(D)$.
For any node set $S \subseteq V(D)$, let $D[S]$ denote the subdigraph of~$D$ induced by $S$.
For a digraph $D$, let $\mathsf{mfvs}(D)$ denote the minimum size of a directed feedback vertex set of~$D$.
And for a digraph $D$ whose nodes are weighted by some function $w:V(D)\rightarrow\mathbb{Q}_{\ge 0}$, let $\mathsf{mfvs}(D,w)$ denote the minimum weight of a directed feedback vertex set of $D$.
Throughout, \emph{cycle} and \emph{triangle} mean directed cycle and directed triangle, respectively.
	
The objective of the {\sc Directed Feedback Vertex Set} problem is to intersect all directed cycles by a minimum-sized set of nodes.
Recall the well-known fact that in tournaments it suffices to hit all directed triangles in order to hit all directed cycles.
	
Another simple but useful observation is that any directed feedback vertex set of~$H(T)$ is a directed feedback vertex set of $T$.
  Thus, to find a directed feedback vertex set of $T$, it suffices to focus on the subdigraph $H(T)$ of $T$.

\begin{proposition}[\cite{GhorbaniM2026}]
\label{lem:H(D)perfect}
  For any tournament $T$, its triangle digraph $H(T)$ is a perfect digraph. 
\end{proposition}

Recall that an acyclic tournament is also called transitive, since its arc set induces a transitive relation.
Every transitive tournament has a unique structure: its nodes admit a unique total ordering.
We denote by $\TT_n$ the transitive tournament on $n$ nodes.
	
\begin{proposition}[\cite{GhorbaniM2026}]
\label{lem:2^{k-1}+1}
	Let $k \geq 2$, and let $T$ be a tournament such that $H(T)$ contains~$\TT_{2^{k-1}+1}$.
	Then~$T$ contains a subtournament from $\T_{2k+1}$.
\end{proposition}

Based on these two propositions, we can prove that for every $k\ge1$, there exists a constant $t=t_k$ such that, for every $\T_{2k+1}$-free tournament $T$, the triangle digraph $H(T)$ is properly $t$-colorable.
That is, there exists a function that assigns colors from $\{1,\hdots,t\}$ to the nodes of $H(T)$ such that any two adjacent nodes receive distinct colors. 
Note that here we crucially refer to the node coloring of the (underlying undirected) graph of $H(T)$.

\theoremstyle{plain}\newtheorem*{re-t(k)-bound-chi(H)}{\Cref{thm:t(k)-bound-chi(H)} (restated)}
\begin{re-t(k)-bound-chi(H)}
Let $k\ge 1$, and let $T$ be a $\T_{2k+1}$-free tournament. 
Then $H(T)$ admits a proper $t$-coloring where $t=t_k\coloneqq2^{2^{k-1}}-1$.	
\end{re-t(k)-bound-chi(H)}
\begin{proof}
  Let $\omega(H(T))$ be the clique number of $H(T)$.
  Since $H(T)$ is perfect (by \Cref{lem:H(D)perfect}), its chromatic number equals its clique number~\cite{Lovasz1972}.
  Thus, it suffices to prove that \mbox{$\omega(H(T))\le 2^{2^{k-1}}-1$}.
			
  Set	$r\coloneqq2^{k-1}+1$.
  We use the Erd\H{o}s--Moser bound~\cite{ErdosMoser1964} for tournaments, which states that every tournament on at least $2^{r-1}$ nodes contains the transitive subtournament $\TT_r$.
	
  Suppose, for sake of contradiction, that $\omega(H(T))\ge 2^{2^{k-1}}$.
  Since $r=2^{k-1}+1$, this means $\omega(H(T))\ge 2^{r-1}$.
  Hence, $H(T)$ contains a clique $Q$ with $|Q|\ge 2^{r-1}$.
  Since $T[Q]$ is a tournament on at least~$2^{r-1}$ nodes, the Erd\H{o}s--Moser bound implies that $T[Q]$ contains a $\TT_r$. 
  Because $Q$ is a clique of $H(T)$, this~$\TT_r$ is contained in $H(T)$.	
  By \Cref{lem:2^{k-1}+1}, this implies that $T$ contains a tournament from~$\T_{2k+1}$.	
  This contradicts the assumption that $T$ is $\T_{2k+1}$-free, and proves the theorem.
\end{proof}

\section{A dynamic programming algorithm for node-weighted directed feedback vertex set on colored triangle digraphs}
\label{sec:dynamic-programming}
Let $(T,w)$ be a node-weighted tournament, and let $C_1,\ldots,C_t$ be the color classes of a proper $t$-coloring of $H(T)$.
Since $V(H(T))=V(T)$, the sets $C_1,\ldots,C_t$ form a partition of $V(T)$.
Moreover, each $C_i$ is an independent set in $H(T)$, so no two nodes of $C_i$ belong to a common directed triangle of~$T$.
Consequently, the subtournament $T[C_i]$ contains no directed triangle and is therefore transitive.

\subsection{Algorithm description}
The following lemma is a key ingredient of our approach.
\begin{lemma}
\label{lem:prefix}
  Let $C$ be an independent set in $H(T)$ and let $c_{1},c_{2},\dots,c_{m}$ be the transitive order of the nodes of $T[C]$, so that $c_{p}\to c_{q}$ whenever $p<q$.
  Then for every node $v\notin C$, the set	
  \begin{equation*}
    N^-_{C}(v)=\big\{c_j\in C: (c_j, v)\in A(T)\big\}
  \end{equation*}
  is a prefix of the ordering of $T[C]$.
  Equivalently, there is an integer $p(v)\in\{0,1,\dots,m\}$ such that	
  \begin{equation*}
    (c_j, v)\in A(T) \quad\hbox{if and only if}\quad j\le p(v) \enspace .
  \end{equation*}
\end{lemma}
\begin{proof}
  Suppose, for sake of contradiction, that $N^-_{C}(v)$ is not a prefix of the transitive ordering of~$T[C]$.
  Then there are indices $a,b$ with $a<b$ such that	$v\to c_{a}$ and	$c_{b}\to v$.
  Since $a<b$ and~$C$ is ordered transitively, $c_{a}\to c_{b}$.
  Therefore, $c_{a}\to c_{b}\to v\to c_{a}$	is a directed triangle in $T$.
  It follows that $\{c_{a},c_{b}\}$ must be an edge of $H(T)$.
  But both $c_{a}$ and $c_{b}$ belong to the independent set $C$ in $H(T)$---a contradiction.	
  Thus, $N^-_{C}(v)$ must be a prefix of the ordering of $C$.
\end{proof}

Now consider a transitive subtournament $U$ of $T$, and suppose that $v$ is the last node in the transitive ordering of $U$.
Then every other node of $U$ must point into $v$.
Hence, by \Cref{lem:prefix}, for each color class $C_i$, the set $U\cap C_i$ is contained in a prefix of the ordering of $C_i$.
This imposes a strong structural constraint on $U$: once the last node is fixed, the remaining nodes of~$U$ can only be chosen from prefixes of the color classes.
Consequently, it suffices to consider subtournaments induced by unions of such prefixes.
We now introduce the notation needed to formalize this idea.
 
Suppose that the transitive order of each $C_i$ is $C_i=\{c_{i,1},c_{i,2},\dots,c_{i,m_i}\}$, where $c_{i,p}\to c_{i,q}$ whenever $p<q$.
For every color class $C_i$ and every node $v\in V(T)$, let $p_i(v)$ be the length of the prefix of $C_i$ that points into $v$.
By \Cref{lem:prefix}, this prefix is well-defined.
More precisely,
\begin{equation*}
  p_i(v)\coloneqq\max\big\{r\in\{0,1,\dots,m_i\}: c_{i,j}\to v\text{ for }j = 1,\hdots,r\big\} \enspace .
\end{equation*}
Note that if $v=c_{i,j}\in C_i$, then $p_i(v)=j-1$.

For vectors $\ba=(a_1,\dots,a_t)$ satisfying $0\le a_i\le m_i$ (these vectors constitute the states of our dynamic program), define
\begin{equation*}
  B(\mathbf a)\coloneqq\bigcup_{i=1}^t\{c_{i,1},c_{i,2},\dots,c_{i,a_i}\} \enspace .
\end{equation*}
Let $M(\ba)$ denote the maximum weight of a transitive subtournament contained in $T[B(\ba)]$, and set $M(\mathbf0)=0$.
Then $M(\mathbf m)$ is the maximum weight of a transitive subtournament of $T$, where 
\begin{equation*}
  \mathbf m\coloneqq(m_1,\dots,m_t) \enspace .
\end{equation*}
For every state  $\ba=(a_1,\dots,a_t)$ and a node $v\in B(\ba)$, define
\begin{equation*}
  \bq^{\ba}(v)\coloneqq(q_1,\dots,q_t)
\end{equation*}
by
\begin{equation*}
  q_i\coloneqq\min\{a_i,p_i(v)\}\qquad\text{for every }i=1,\dots,t \enspace .
\end{equation*}

The following lemma allows us to compute $M(\ba)$ recursively for every state $\ba$, and in particular to compute the maximum weight of a transitive subtournament of $T$.
\begin{lemma}
  For every state $\ba=(a_1,\dots,a_t)$ it holds
  \begin{equation}
  \label{eq:recursion}
     M(\ba)	=\max_{v\in B(\ba)}	\bigl(	w(v)+M(\bq^{\ba}(v))\bigr) \enspace .
  \end{equation}
\end{lemma}
\begin{proof}
  We prove the assertion by induction on
  \begin{equation*}
	  |\ba|\coloneqq a_1+\cdots+a_t \enspace .
  \end{equation*}
  If $\ba=\mathbf0$, then $B(\ba)=\emptyset$, and the maximum weight of a transitive subtournament contained in $T[B(\ba)]$ is $0$.
  Hence, the recurrence is correct in the base case.	
  Now suppose that $\ba\neq \mathbf0$, and assume the recurrence is correct for all states $\mathbf b$ with $|\mathbf b|<|\ba|$.
	
  Let $S\subseteq B(\ba)$ be a maximum-weight transitive subtournament of $T[B(\ba)]$.	
  Let $v$ be the last node of $S$ in its transitive ordering.
  Then every node of $S\setminus\{v\}$ points to $v$.
  By \Cref{lem:prefix}, for each color class $C_i$, the nodes of $C_i$ that point into $v$ form a prefix of the transitive ordering of $C_i$ with length $p_i(v)$.
  Since $S\subseteq B(\ba)$, the nodes of $S\setminus\{v\}$ belonging to $C_i$ must lie among $c_{i,1},c_{i,2},\dots,c_{i,q_i}$ where $q_i=\min\{a_i,p_i(v)\}$.
  Therefore, $S\setminus\{v\}\subseteq B(\bq^{\ba}(v))$.
  Hence,
  \begin{equation*}
	  w(S)=w(v)+w(S\setminus\{v\})\le w(v)+M(\bq^{\ba}(v)) \enspace .
  \end{equation*}
  Since $S$ was optimal, this shows that $M(\ba)\le \max_{v\in B(\ba)} \bigl(w(v)+M(\bq^{\ba}(v))\bigr)$.
	
  Conversely, fix any node $v\in B(\ba)$, and let $S'$ be a maximum-weight transitive subtournament contained in $T[B(\bq^{\ba}(v))]$.
  By definition of $\bq^{\ba}(v)$, every node of $S'$ points to $v$.	
  Since $T[S']$ is transitive, appending $v$ as the last node gives a transitive subtournament $T[S'\cup\{v\}]$.
  Its weight is $w(S')+w(v) = M(\bq^{\ba}(v))+w(v)$.
  Therefore, every candidate value $w(v)+M(\bq^{\ba}(v))$	appearing in the recurrence is attainable by a transitive subtournament of $T[B(\ba)]$.
  Hence, $M(\ba)\ge \max_{v\in B(\ba)} \bigl(w(v)+M(\bq^{\ba}(v))\bigr)$.
\end{proof}

The recursion \eqref{eq:recursion} enables us to find a maximum-weight transitive subtournament, and hence a minimum-weight directed feedback vertex set, by dynamic programming.
For each state $\ba$, we store a node $v$ for which the maximum in \eqref{eq:recursion} is attained.
We call these nodes \emph{predecessor pointers}.
More precisely, when computing $M(\ba)$, we store a node
\begin{equation*}
  \operatorname{pred}(\ba)\in B(\ba)
\end{equation*}
that attains the maximum in the recurrence.
These pointers are then used to reconstruct an optimal transitive subtournament.
The formal description is given in \Cref{alg:triangle-graph-dp}.

\begin{algorithm}[h!]
	\caption{Exact dynamic program for \dfvs from a proper $t$-coloring of the triangle digraph}
	\label{alg:triangle-graph-dp}	
	\begin{algorithmic}[1]
		\Require A tournament $T$ with node weights $w:V(T)\to \mathbb{Q}_{\ge 0}$, and color classes $C_1,\ldots,C_t$ of a proper $t$-coloring of~$H(T)$.
		\Ensure A minimum-weight directed feedback vertex set $F$ of $(T,w)$.
		
		\For{$i=1,\dots,t$}
		\State Compute the transitive order of $T[C_i]$:
		$C_i=\{c_{i,1},c_{i,2},\dots,c_{i,m_i}\}$
		such that $c_{i,p}\to c_{i,q}$ whenever $p<q$.
		
		\For{each node $v\in V(T)$}
		\State Set
		$
		p_i(v)\gets
		\max\{r\in\{0,1,\dots,m_i\}: (c_{i,j}, v)\in A(T) \text{ for } j=1,\hdots,r\}.
		$
		\EndFor
		\EndFor
		
		
		\State Set $M(\mathbf0)\gets 0$. 
		
		\For{$s=1,\dots,n$}
		\For{all states $\ba=(a_1,\dots,a_t)$ with $a_1+\cdots+a_t=s$}
		\State Set $M(\ba)\gets 0$. 
		\State Define $B(\ba) = \bigcup_{i=1}^t \{c_{i,1},c_{i,2},\dots,c_{i,a_i}\}.$
		\For{each node $v\in B(\ba)$}
		\State Define $\bq^{\ba}(v)=(q_1,\dots,q_t)$ by $q_i\gets \min\{a_i,p_i(v)\}$ for every $i=1,\dots,t$.
		\State $\operatorname{value}\gets w(v)+M(\bq^{\ba}(v))$.
		\If{$\operatorname{value}>M(\ba)$}
		\State $M(\ba)\gets \operatorname{value}$.
		\State $\operatorname{pred}(\ba)\gets v$.
		\EndIf
		\EndFor
		\EndFor
		\EndFor
		
		\State $U\gets \emptyset$.
		\State $\ba\gets \mathbf m$.
		\While{$\ba\neq \bf0$}
		\State $v\gets \operatorname{pred}(\ba)$.
		\State $U\gets U\cup\{v\}$.
		\State  $\ba\gets\bq^{\ba}(v)$.
		\EndWhile
		
		\State $F\gets V(T)\setminus U$.
		\State \Return $F$.
	\end{algorithmic}
\end{algorithm}

\subsection{Algorithm analysis}
\theoremstyle{plain}\newtheorem*{re-dynamic-program}{\Cref{thm:dynamic-program} (restated)}
\begin{re-dynamic-program}
  Given a node-weighted tournament $(T,w)$ of order $n$ and a proper\linebreak $t$-coloring of $H(T)$, \Cref{alg:triangle-graph-dp} computes a minimum-weight directed feedback vertex set of $(T,w)$ in time $O(tn^{t+1})$.
\end{re-dynamic-program}
\begin{proof}
  We prove correctness of the algorithm, and then analyze its asymptotic run time.
  
  If $t\le2$, then $T$ contains no directed triangle, and so there is nothing to prove.
  Thus we assume that $t\ge3$.
	
  Minimizing the weight of a feedback vertex set of $(T,w)$ is equivalent to maximizing the weight of a transitive subtournament of $(T,w)$.
  Therefore, it is enough to show that the dynamic program computes a maximum-weight transitive subtournament $U$ of $T$.
	
  The recurrence \eqref{eq:recursion} computes $M(\ba)$ for every state $\ba$.
  It is well-defined in the sense that, if the states are processed in increasing order of $|\ba|$, then all entries needed to compute $M(\ba)$ have already been computed.
  To see this, suppose that $v=c_{h,j}\in B(\ba)$.
  Then $j\le a_h$, and by construction $p_h(v)=j-1$.
  Therefore $q_h=\min\{a_h,p_h(v)\}=j-1<a_h$.
  For every $i\neq h$, we have $q_i\le a_i$.
  Hence, $|\bq^{\ba}(v)|<|\ba|.$
	
  Therefore, after the  table is filled, $M(\mathbf m)$ is the maximum weight of a transitive subtournament of $T$.
  Following the stored predecessor pointers reconstructs such a maximum-weight transitive subtournament $U$.
  Then $F=V(T)\setminus U$ is a minimum-weight feedback vertex set of $T$.
	
  Now we analyze the run time.
  There are $\prod_{i=1}^t(m_i+1)$ states, because each coordinate $a_i$ can take one of $m_i+1$ values.	
  For each state $\ba$, the recurrence considers all nodes in $B(\ba)$.
  Since $|B(\ba)|\le n$, there are at most $n$ candidate last nodes.	
	For each candidate node $v$, computing $\bq^{\ba}(v)$ takes $O(t)$ time if done directly, because it requires computing $q_i=\min\{a_i,p_i(v)\}$ for every $i=1,\dots,t$.
  Thus, the total dynamic programming time is $O\left(t n\prod_{i=1}^t(m_i+1)\right)$.
	
  The preprocessing takes polynomial time.
  The triangle digraph $H(T)$ can be constructed na{\"i}vely in $O(n^3)$ time by checking every triple of nodes and marking the three pairs whenever the triple forms a directed triangle.
  Given the color classes, each $T[C_i]$ can be ordered transitively in $O(m_i^2)$ time, and therefore all color classes can be ordered in time $O\left(\sum_{i=1}^t m_i^2\right)\le O(n^2)$.
  The values $p_i(v)$ can be computed by scanning each ordered color class for each node.
  This takes $O\left(\sum_{i=1}^t n m_i\right)=O(n^2)$ time.
  Therefore, the overall run time is
  \begin{equation*}
	  O\left(n^3+t n\prod_{i=1}^t(m_i+1)\right)=O\left(n^3+t n\cdot n^t\right)=O(tn^{t+1}) \enspace .\qedhere
  \end{equation*}
\end{proof}

\begin{proof}[\bf Proof of \Cref{cor:poly-DFVS-Tk-free}]
  For $k=1$, there is nothing to prove, as every $\T_3$-free tournament has an empty directed feedback vertex set.
  Let $k\geq 2$ be fixed, and let $T$ be a $\T_{2k+1}$-free tournament.
  By \Cref{thm:t(k)-bound-chi(H)}, there is a constant $t=t_k$ such that its triangle digraph $H(T)$ admits a proper $t$-coloring, and by the perfectness of $H(T)$, such a coloring can be found in polynomial time by an application of the ellipsoid algorithm, as explained by Gr{\"o}tschel, Lov{\'a}sz, and Schrijver~\cite{GrotschelLS1984}.
  So, by \Cref{thm:dynamic-program}, a minimum-weight directed feedback vertex set of $(T,w)$ can be computed in time~$O(tn^{t+1})$.
\end{proof}
We remark that our algorithm can easily be made fully combinatorial, and that the ellipsoid method can be avoided.
Ghorbani and Mnich~\cite{GhorbaniM2026} showed that $H(T)$ is a cocomparability graph.
So its complement $\overline{H(T)}$ defines a poset $P$, whose chains are independent sets of $H(T)$.
By Dilworth's theorem, a minimum chain decomposition, and hence an optimal coloring of $H(T)$, can be found in polynomial time via maximum bipartite matching.
Thus, the coloring step can be implemented combinatorially.

\section{A deterministic $(2+\varepsilon)$-approximation algorithm for node-weighted directed feedback vertex set in tournaments}
\label{sec:2+epsilon-appx}
We now present our deterministic $(2+\e)$-approximation algorithm, \Cref{alg:weighted-local-ratio}, for the node-weighted {\sc Directed Feedback Vertex Set} problem in tournaments.

\begin{algorithm}[t]
	\caption{Weighted local-ratio algorithm for \dfvs}
	\label{alg:weighted-local-ratio}
	\begin{algorithmic}[1]
		\Require A tournament $T$ with node weights $w:V(T)\to\mathbb Q_{\ge 0}$, and an integer $k\ge 1$.
		\Ensure A directed feedback vertex set $F_{\mathrm{out}}$ of $T$ with weight at most $\left(2+\nicefrac1k\right)\mfvs(T,w)$.
		
		\State $R\gets V(T)$.
		\State $F_0\gets \emptyset$.
		\State Let $\widetilde w$ be the current weight function, initially $\widetilde w=w$.
		
		\While{there exists a node $v\in R$ with $\widetilde w(v)=0$}
		\State $F_0\gets F_0\cup\{v\}$.
		\State $R\gets R-v$.
		\EndWhile
		
		\While{$R$ contains a subset $S$ such that $T[S]\in{\cal T}_{2k+1}$}\label{line:S-definition}
		\State $\delta\gets \min\{\widetilde w(v):v\in S\}$.
		\For{each $v\in S$}
		\State $\widetilde w(v)\gets \widetilde w(v)-\delta$.
		\EndFor
		
		\While{there exists a node $v\in R$ with $\widetilde w(v)=0$}
		\State $F_0\gets F_0\cup\{v\}$.
		\State $R\gets R-v$.
		\EndWhile
		\EndWhile
		
		\State $T'\gets T[R]$.
		\State Let $F'$ be the output of \Cref{alg:triangle-graph-dp} on $(T',\widetilde w)$.
		\State $F_{\mathrm{out}}\gets F_0\cup F'$.
		\State \Return $F_{\mathrm{out}}$.
	\end{algorithmic}
\end{algorithm}

In the first phase, \Cref{alg:weighted-local-ratio} employs the local-ratio technique.
The input consists of a tournament~$T$ together with a nonnegative node-weight function $w:V(T)\to\mathbb{Q}_{\ge 0}$.
The algorithm repeatedly searches for a subset $S\subseteq V(T)$ such that $T[S]\in\T_{2k+1}$.
Whenever such a subset is found, it subtracts the same amount from the weights of all nodes in $S$, thereby obtaining an updated weight function $\widetilde{w}$.
If the weight of a node $v$ becomes zero, that node is placed into a set~$F_0$, while the remaining nodes form a set $R$.
Once all subtournaments belonging to $\T_{2k+1}$ have been eliminated, the residual tournament $T' = T[R]$ is $\T_{2k+1}$-free.
We can then apply \Cref{alg:triangle-graph-dp} to compute an exact minimum-weight directed feedback vertex set of $T'$.
The approximation ratio of $(2k+1)/k$ follows from the local-ratio phase.
Each iteration processes a subtournament on $2k+1$ nodes, while every directed feedback vertex set must contain at least $k$ nodes from that subtournament.
Consequently, each local-ratio step incurs a factor of at most $(2k+1)/k$.
By choosing~$k$ so that $1/k\le\varepsilon$, we obtain our main result, which we restate below.

For the analysis, we introduce one further notation: given a weight function $z$ on $V(T)$ and a subset $S \subseteq V(T)$, we write $z(S)$ to denote $\sum_{v \in S} z(v)$.

\theoremstyle{plain}\newtheorem*{re-2+epsilon-tournament}{\Cref{thm:2+epsilon-tournament} (restated)}
\begin{re-2+epsilon-tournament}
  Let $k\ge 1$ be fixed.
  For any input tournament $(T,w)$, \Cref{alg:weighted-local-ratio} returns in polynomial time a directed feedback vertex set $F_{\mathrm{out}}$ of $T$ satisfying
  \begin{equation*}
    w(F_{\mathrm{out}})\le \left(2+\nicefrac1k\right)\mfvs(T,w) \enspace .
  \end{equation*}
\end{re-2+epsilon-tournament}
\begin{proof}
  We first verify that the output of \Cref{alg:weighted-local-ratio} is a directed feedback vertex set for the input tournament $T$.
  During the run of algorithm, nodes placed in $F_0$ are deleted from the current residual tournament.
  At the end of the run, the remaining tournament is $T'$.
  The set $F'$ is a directed feedback vertex set of $T'$, computed exactly by \Cref{alg:triangle-graph-dp}.
  Therefore, $T-(F_0\cup F')$	is acyclic.
  Hence, $F_{\mathrm{out}}=F_0\cup F'$ is a directed feedback vertex set of $T$.
	
  Now we prove the approximation guarantee.
  Each time the algorithm finds an $S_j\subseteq R$ with $T[S_j]\in{\cal T}_{2k+1}$, it defines $\delta_j\coloneqq\min\{\widetilde w(v):v\in S_j\}$ and subtracts $\delta_j$ from every node of $S_j$.
  Thus, the original weight function $w$ can be decomposed as
  \begin{equation*}
	  w=w^*+\sum_{j=1}^m \delta_j \mathbf 1_{S_j},
  \end{equation*}
  where $w^*$ is the final residual weight function, extended by zero to all nodes removed during the run of the algorithm, and $\mathbf 1_{S_j}$ is the indicator weight function of $S_j$.
	
  We use the local-ratio principle.
  Let $\alpha\coloneqq2+\nicefrac{1}{k}$.
  It is enough to show that the output $F_{\mathrm{out}}$ is an $\alpha$-approximate solution with respect to each component weight function $w^*$ and $\delta_j\mathbf 1_{S_j}$.
  Let $Z$ be an arbitrary directed feedback vertex set of the input tournament $T$.
	
  First consider a local component $\delta_j\mathbf 1_{S_j}$.
  On the one hand, since $T[S_j]\in{\cal T}_{2k+1}$, every directed feedback vertex set of $T[S_j]$ has size at least $k$.
  Therefore, $|Z\cap S_j|\ge k$.
  Hence,
  \begin{equation*}
	  \delta_j\mathbf 1_{S_j}(Z)=\delta_j |Z\cap S_j|\ge\delta_j k \enspace .
  \end{equation*}
  On the other hand, the algorithm's output can contain at most all nodes of $S_j$.
  Since $|S_j|=2k+1$, we have	
  \begin{equation*}
	  \delta_j\mathbf 1_{S_j}(F_{\mathrm{out}})	=   \delta_j |F_{\mathrm{out}}\cap S_j|
                                              \le \delta_j(2k+1) \enspace .
  \end{equation*}
  Therefore,
  \begin{equation*}
	\delta_j\mathbf 1_{S_j}(F_{\mathrm{out}}) \le	\frac{2k+1}{k}\,	\delta_j\mathbf 1_{S_j}(Z)
	                                            = \alpha\,\delta_j\mathbf 1_{S_j}(Z) \enspace .
  \end{equation*}
  Thus, $F_{\mathrm{out}}$ is an $\alpha$-approximate solution with respect to every local component $\delta_j\mathbf 1_{S_j}$.
	
  Now consider the final residual weight function $w^*$.
  By construction, $w^*$ is zero on all nodes deleted into $F_0$, and equals the final weights $\widetilde w$ on $V(T')$.
  Therefore, $w^*(F_{\mathrm{out}})=w^*(F')$.
  The set~$F'$ is a minimum-weight directed feedback vertex set of $T'$ with respect to the weights $\widetilde w$.  
  We have that $Z\cap V(T')$ is a directed feedback vertex set of $T'$, because any directed cycle in $T'-(Z\cap V(T'))$	would also be a directed cycle in $T-Z$.
  Therefore, by optimality of $F'$ in $T'$, it holds that
  \begin{equation*}
    w^*(F')\le w^*(Z\cap V(T'))\le	w^*(Z) \enspace .
  \end{equation*}
  Hence,
  \begin{equation*}
	  w^*(F_{\mathrm{out}})	\le	w^*(Z)	\le	\alpha w^*(Z) \enspace .
  \end{equation*}
  Thus, the output is an $\alpha$-approximate solution with respect to $w^*$ as well.
  Combining the inequalities over the weight decomposition gives
  \begin{equation*}
	  w(F_{\mathrm{out}}) =   w^*(F_{\mathrm{out}}) + \sum_{j=1}^m \delta_j\mathbf 1_{S_j}(F_{\mathrm{out}})
		                    \le \alpha w^*(Z) + \alpha \sum_{j=1}^m \delta_j\mathbf 1_{S_j}(Z)
	                      =   \alpha w(Z) \enspace .
  \end{equation*}
  Since this relation holds for every directed feedback vertex set $Z$, it holds in particular for any minimum-weight directed feedback vertex set.
  Therefore, $w(F_{\mathrm{out}}) \le \alpha\, \mfvs(T,w)$.
  This establishes the approximation factor.
	
  It remains to analyze the run time of the algorithm.
  For fixed $k$, membership in ${\cal T}_{2k+1}$ can be tested in constant time with respect to $n$, because the subtournament has constant size~$2k+1$.
  A forbidden subtournament can be found by enumerating all subsets of size $2k+1$, which takes time $O(n^{2k+1})$ per search.	
  Each iteration of the outer loop creates at least one new zero-weight node, which is then removed from the residual tournament.
  Hence, the number of weight reduction iterations is at most $n$.	
  Thus the overall time spent searching for forbidden subtournaments is $O(n^{2k+2})$.	
  After the loop terminates, the residual tournament $T'$ is ${\cal T}_{2k+1}$-free.
  A proper $t_k$-coloring of $H(T')$ can be computed in polynomial time for fixed $k$, by the ellipsoid algorithm~\cite{GrotschelLS1984}.
  Applying \Cref{alg:triangle-graph-dp} on $T'$ with $t=t_k$ color classes runs in time~$n^{O(t)}$.
  Therefore, for fixed~$k$, the total run time is polynomial in~$n$.
\end{proof}

\begin{proof}[\bf Proof of \Cref{cor:dfvs-appx-quasi-transitive}]
  Let $D$ be a quasi-transitive digraph with non-negative node weights\linebreak $w:V(D)\to\mathbb Q_{\ge 0}$.	
  Ghorbani and Mnich~\cite[Lemma~4]{GhorbaniM2026} established that there exists a tournament~$T$ with $V(T)=V(D)$ such that $D$ is a subdigraph of $T$ and $H(D)=H(T)$.
	
  We claim that a set $F\subseteq V(D)$ is a directed feedback vertex set of $D$ if and only if it is a directed feedback vertex set of $T$.
  Indeed, every directed cycle in a tournament contains a directed triangle.
  The same is true for quasi-transitive digraphs: if a shortest directed cycle had length at least four, quasi-transitivity applied to two consecutive arcs would either create a directed triangle or a shorter directed cycle, a contradiction.
  Hence, in both $D$ and $T$, a node set is a directed feedback vertex set if and only if it intersects every directed triangle.
	
  Since $H(D)=H(T)$, the directed triangles relevant to the {\sc Directed Feedback Vertex Set} problem are preserved.
  Therefore, for every $F\subseteq V(D)=V(T)$, the set $F$ is a directed feedback vertex set of $D$ if and only if $F$ is a directed feedback vertex set of $T$.
  In particular, $\mfvs(D,w)=\mfvs(T,w)$, and the output $F_{\rm out}$ of \Cref{alg:weighted-local-ratio} is a directed feedback vertex set of~$D$.
  So the result follows.
\end{proof}

\section{Improving the run time of \Cref{alg:weighted-local-ratio}}\label{sec:runtime-improve}
The run time of \Cref{alg:weighted-local-ratio} is dominated by \Cref{alg:triangle-graph-dp} applied to the residual tournament.
In the previous analysis, the residual tournament is $\T_{2k+1}$-free, and hence the clique number of its triangle digraph is bounded by $2^{2^{k-1}}-1$.
This yields a run time of $O(n^{2^{2^{k-1}}})$, up to factors depending only on $k$.
In this section we show how to reduce the doubly-exponential dependence of the run time on $k$ to a single-exponential dependence in $k$.

We can modify the local-ratio phase in \Cref{alg:weighted-local-ratio} so that the triangle digraph of the residual tournament has clique number at most $2^k$.
This property is achieved by excluding, in addition to copies from $\T_{2k+1}$, certain larger subtournaments with the same ratio $\mfvs(T[S])/|S|$.
More precisely, we find subsets $S$ of $V(T)$ such that
\begin{equation*}
  |S|\in\big\{2k+1,2^k\big\}\quad\hbox{and}\quad \frac{\mfvs(T[S])}{|S|}\ge\frac{k}{2k+1} \enspace .
\end{equation*}
This modification only changes Line~\ref{line:S-definition} of \Cref{alg:weighted-local-ratio}; the rest of the algorithm remains unchanged.
We call the resulting algorithm the \emph{modified} \Cref{alg:weighted-local-ratio}.

The modification preserves the approximation guarantee.
Indeed, exactly as in the proof of \Cref{alg:weighted-local-ratio}, every feedback vertex set intersects $S$ in at least $\mfvs(T[S])$ nodes, and it holds that $\mfvs(T[S])/|S|\ge k/(2k+1)$.
Therefore, the usual local-ratio proof applies unchanged, and the modified algorithm still returns a deterministic
$\left(2+\nicefrac1k\right)$-approximation for $\mfvs(T,w)$.

\begin{proposition}
  For fixed $k\ge1$, the run time of the modified \Cref{alg:weighted-local-ratio} is $O(n^{2^k})$.
\end{proposition}
\begin{proof}
  For $k=1$, there is nothing to prove, so assume that $k\ge2$. 

  Let $T'$ be the residual tournament after the modified local-ratio phase.
  Then $T'$ is $\T_{2k+1}$-free.
  We claim that $\omega(H(T'))<2^k$.
  Suppose, for sake of contradiction, that $H(T')$ contains a clique $Q$ of size $2^k$.
  Since $Q$ survives in the residual tournament, the subtournament $T'[Q]$ was not processed by the modified local-ratio phase.
  Hence,
  \begin{equation*}
    \frac{\mfvs(T'[Q])}{|Q|}<\frac{k}{2k+1} \enspace .
  \end{equation*}
  Thus $T'[Q]$ contains a transitive subtournament of size
  \begin{equation*}
    |Q|-\mfvs(T'[Q])>\left(1-\frac{k}{2k+1}\right)|Q|=\frac{k+1}{2k+1}\,2^k>2^{k-1} \enspace .
  \end{equation*}
  Hence $T'[Q]$ contains $\TT_{2^{k-1}+1}$.
  Moreover, since $Q$ is a clique in $H(T')$, $\TT_{2^{k-1}+1}$ lies in $H(T')$.
  By \Cref{lem:2^{k-1}+1}, $T'$ contains a subtournament from $\T_{2k+1}$, contradicting the fact that $T'$ is $\T_{2k+1}$-free.
  Therefore, $\omega(H(T'))\le2^k-1$.
  Since $H(T')$ is perfect, it is $(2^k-1)$-colorable, and hence by \Cref{thm:dynamic-program}, \Cref{alg:triangle-graph-dp} solves the residual weighted instance in time $O(n^{2^k})$. 
  
  It remains to justify the run time of the modified local-ratio phase.
  For fixed $k\ge 3$, we have $2k+1\le 2^k$.
  We enumerate all subtournaments of sizes $2k+1$ and $2^k$ once.
  For each such subtournament, the condition $\mfvs(T[S])/|S|\ge k/(2k+1)$ depends only on the induced tournament and can be
  checked in constant time depending only on $k$.
  During the local-ratio phase, a precomputed candidate $S$ is processed if all its nodes are still present in
  the residual tournament, and otherwise it is skipped.
  Since the residual node set only decreases, one pass over all candidates is sufficient.
  Thus, for fixed $k\ge 3$, the local-ratio phase takes time $O(n^{2^k})$.
  
  For $k=2$, the $4$-node candidates are enumerated in time $O(n^4)$, whereas subtournaments from~$\T_5$ can be detected in time $O(n^3)$ (because by \Cref{lem:2^{k-1}+1} one only needs to locate copies of $\TT_3$ in $H(T')$).
  Thus the case $k=2$ also runs in time $O(n^4)=O(n^{2^k})$.
\end{proof} 

\section{Concluding remarks}
\label{sec:concludingremarks}
Our main result is the first deterministic polynomial-time $(2+\varepsilon)$-approximation algorithm for the unweighted, and node-weighted, {\sc Directed Feedback Vertex Set} problem in tournaments.
We thereby improve upon a sequence of earlier deterministic polynomial-time approximation algorithms  for this problem with worse approximation ratio, and almost match the approximation factor of the best-known randomized, and quasi-polynomial time, approximation algorithms.


Before stating our final remarks, we introduce a family of tournaments and establish some of its properties.
Let $n=2m+1$ be odd.
The \emph{regular cyclic tournament} $R_n$ is the tournament with vertex set $\{0,\dots,2m\}$ in which $i\to j$ if and only if $j-i\pmod n\in\{1,\dots,m\}$.

For a tournament $T$, an \emph{acyclic $k$-coloring} of $T$ is a coloring of its nodes with colors from the set $\{1,\dots,k\}$ such that each color class induces an acyclic (transitive) subtournament.
In fact the dichromatic number $\vec\chi(T)$ of $T$ is the smallest integer $k$ for which $T$ admits an acyclic $k$-coloring.
\begin{proposition}
\label{prop:regular-cyclic-tournament}	
  The tournament $R_{n}$ has the following properties:
  \begin{itemize}
	  \item[\rm(i)] $H(R_n)=R_n$, 
	  \item[\rm(ii)] $\vec\chi(R_n)=2$, 
	  \item[\rm(iii)] $R_n$ contains no subtournament $S$ with $\mfvs(S)\ge|S|/2$. 
  \end{itemize}
\end{proposition}
\begin{proof} 	(i) It is enough to show that every arc of $R_n$ belongs to a directed triangle.		
	Let $i\to j$ be an arbitrary arc of $R_n$. 
  Then $i\to j\to j+m\to i$ is a directed triangle. (It is easy to verify that $i-(j + m) \pmod n$ belongs to $\{1,\dots,m\}$).
	
	(ii) It suffices to show that $R_n$ can be partitioned into two transitive subtournaments.
	This follows from the stronger property that, in $R_n$, the out-neighborhood and in-neighborhood of any node $v$ are cyclic intervals of length $m$:
\begin{align*}
	N^+_{R_n}(v) & =\{v+1,\dots,v+m\}\hspace{-2.5mm}\pmod{n},\\
	N^-_{R_n}(v) & =\{v-m,\dots,v-1\}\hspace{-2.5mm}\pmod{n}.
\end{align*}
	Note that any cyclic interval of length $m$ induces a transitive subtournament in $R_n$, and hence each of these sets, together with $v$, induces a transitive subtournament.
	
 (iii) From the proof of (ii) it follows that for any $v\in S$ both $\{v\}\cup N_S^+(v)$ and $N_S^-(v)\cup\{v\}$ induce transitive subtournaments of $S$.
   At least one of these has size $1+(|S|-1)/2=(|S|+1)/2$.
   This means $\mfvs(S)<|S|/2.$
\end{proof}

\paragraph{Remarks.}	{\bf1.}
 The approximation ratio $2+\nicefrac{1}{k}$ arises from the local-ratio charges on subtournaments in~$\T_{2k+1}$.
 Indeed, if $Q\in\T_{2k+1}$, then $|Q|/\mfvs(Q)\le 2+\frac{1}{k}$.
It is therefore tempting to replace~$\T_{2k+1}$ by another family $\mathcal F$ of subtournaments satisfying $|Q|/\mfvs(Q)\le 2$ for every $Q\in\mathcal F$, or equivalently $\mfvs(Q)\ge |Q|/2$.
Such a choice would make every local-ratio charge individually compatible with an approximation factor-$2$ analysis.
However, this approach does not provide a useful residual structure for the dynamic programming phase.
The reason is that forbidding subtournaments~$Q$ with $\mfvs(Q)\ge |Q|/2$ does not imply that the triangle digraph of the residual tournament has bounded clique number.
To see this, consider the regular cyclic tournament $R_n$.
By \Cref{prop:regular-cyclic-tournament}, its triangle digraph $H(R_n)$ is $R_n$ itself, and so $\omega(H(R_n))=n$.
Thus the clique number of $H(R_n)$ is unbounded.
On the other hand, $R_n$ contains no subtournament $S$ satisfying \mbox{$\mfvs(S)\ge|S|/2$}.
Therefore, $R_n$ would survive a local-ratio phase that only removes subtournaments~$Q$ with $\mfvs(Q)\ge |Q|/2$, but $H(R_n)$ has clique number $n$.
This shows that the family~$\T_{2k+1}$ is useful not merely because it gives a local-ratio bound, but because forbidding $\T_{2k+1}$ imposes strong structure on the residual tournament: it bounds the clique number of the triangle digraph.

{\bf2.} 
As each color class of a proper coloring of $H(T)$ is independent in $H(T)$, it induces a transitive subtournament of $T$.
Thus one might hope to replace a proper node coloring of $H(T)$ by an acyclic coloring of $T$ and use the color classes of such a coloring in the dynamic program.
Though, this is not sufficient.
The dynamic program does not only require the color classes to be transitive; it also uses the stronger prefix property from \Cref{lem:prefix}.
For an arbitrary acyclic coloring of $T$, the color classes are transitive, but this prefix property need not hold. The set $N_T^-(v)\cap C_i$ can be an arbitrary subset of the transitive order of $C_i$, not necessarily a prefix, and thus the dynamic program no longer applies.

\paragraph{Open question.}
As part of our approach, we give an algorithm for node-weighted {\sc Directed Feedback Vertex Set} on $\mathcal T_{2k+1}$-free tournaments which runs in polynomial time $n^{O(2^k)}$, for every fixed $k$.
This naturally leads to the question whether the problem is fixed-parameter tractable parameterized by $k$ on $\mathcal T_{2k+1}$-free tournaments, and can  thus be solved in time $f(k)\cdot n^{O(1)}$ for some computable function $f$.
A positive answer would be a necessary ingredient to strengthen our approximation algorithm to an algorithm, which would, for every $\varepsilon>0$, compute a $(2+\varepsilon)$-approximation in time $f(1/\varepsilon)\cdot n^{O(1)}$, thereby removing the dependence on $\varepsilon$ from the exponent of the run time.
 
\vspace{.5cm}


	
	
	
	

	\bibliography{bibliography}

\begin{thebibliography}{10}

\bibitem{AprileDFH2023}
Manuel Aprile, Matthew Drescher, Samuel Fiorini, and Tony Huynh.
\newblock A $7/3$-approximation algorithm for feedback vertex set in
  tournaments via {S}herali–{A}dams.
\newblock {\em Discrete Appl. Math.}, 337:149--160, 2023.

\bibitem{AprileDFH2023b}
Manuel Aprile, Matthew Drescher, Samuel Fiorini, and Tony Huynh.
\newblock A tight approximation algorithm for the cluster vertex deletion
  problem.
\newblock {\em Math. Program.}, 197(2):1069--1091, 2023.

\bibitem{BangJensenG2018}
J{\o}rgen Bang-Jensen and Gregory Gutin.
\newblock {\em Classes of directed graphs}, volume~11.
\newblock Springer, 2018.

\bibitem{BangJensenH1995}
Jørgen Bang-Jensen and Jing Huang.
\newblock Quasi-transitive digraphs.
\newblock {\em J. Graph Theory}, 20(2):141--161, 1995.

\bibitem{BarYehudaBFR2004}
Reuven Bar-Yehuda, Keren Bendel, Ari Freund, and Dror Rawitz.
\newblock Local ratio: A unified framework for approximation algorithms. in
  memoriam: Shimon {E}ven 1935-2004.
\newblock {\em ACM Comput. Surv.}, 36(4):422--463, December 2004.

\bibitem{BergerCCFLSST2013}
Eli Berger, Krzysztof Choromanski, Maria Chudnovsky, Jacob Fox, Martin Loebl,
  Alex Scott, Paul Seymour, and St\'ephan Thomass\'e.
\newblock Tournaments and colouring.
\newblock {\em J. Combin. Theory Ser. B}, 103(1):1--20, 2013.

\bibitem{CaiDZ1998}
Mao-cheng Cai, Xiaotie Deng, and Wenan Zang.
\newblock A {TDI} system and its application to approximation algorithms.
\newblock In {\em Proc. FOCS 1998}, page 227, 1998.

\bibitem{CaiDZ2001}
Mao-Cheng Cai, Xiaotie Deng, and Wenan Zang.
\newblock An approximation algorithm for feedback vertex sets in tournaments.
\newblock {\em SIAM J. Comput.}, 30(6):1993--2007, 2001.

\bibitem{ErdosMoser1964}
Paul Erd{\H{o}}s and Leo Moser.
\newblock A problem on tournaments.
\newblock {\em Canadian Mathematical Bulletin}, 7(3):351--356, 1964.

\bibitem{FioriniJS2016}
Samuel Fiorini, Gwena{\"e}l Joret, and Oliver Schaudt.
\newblock Improved approximation algorithms for hitting 3-vertex paths.
\newblock In {\em Proc. IPCO 2016}, pages 238--249, 2016.

\bibitem{FioriniJS2020}
Samuel Fiorini, Gwena{\"e}l Joret, and Oliver Schaudt.
\newblock Improved approximation algorithms for hitting 3-vertex paths.
\newblock {\em Math. Program.}, 182(1):355--367, 2020.

\bibitem{GhorbaniM2026}
Ebrahim Ghorbani and Matthias Mnich.
\newblock A $9/4$-approximation for directed feedback vertex sets in
  quasi-transitive digraphs.
\newblock In {\em Proc. ICALP 2026}, Leibniz Int. Proc. Informatics, pages
  151:1--151:16, 2026.
\newblock Article No. 151.

\bibitem{GrotschelLS1984}
Martin Gr{\"o}tschel, L{\'a}szl{\'o} Lov{\'a}sz, and Alexander Schrijver.
\newblock Polynomial algorithms for perfect graphs.
\newblock {\em Annals Discrete Math.}, 21:325--356, 1984.

\bibitem{GuruswamiL2014}
Venkatesan Guruswami and Euiwoong Lee.
\newblock Inapproximability of feedback vertex set for bounded length cycles.
\newblock In {\em Electronic colloquium on computational complexity (ECCC)},
  volume~21, page~2, 2014.

\bibitem{KhotR2008}
Subhash Khot and Oded Regev.
\newblock Vertex cover might be hard to approximate to within $2-\varepsilon$.
\newblock {\em J. Comput. Syst. Sci.}, 74(3):335--349, 2008.

\bibitem{LokshtanovK2016}
Mithilesh Kumar and Daniel Lokshtanov.
\newblock Faster exact and parameterized algorithm for feedback vertex set in
  tournaments.
\newblock In {\em Proc. STACS 2016}, volume~47 of {\em Leibniz Int. Proc.
  Informatics}, pages 49:1--49:13, 2016.

\bibitem{LokshtanovMMPPS2020}
Daniel Lokshtanov, Pranabendu Misra, Joydeep Mukherjee, Fahad Panolan,
  Geevarghese Philip, and Saket Saurabh.
\newblock 2-approximating feedback vertex set in tournaments.
\newblock In {\em Proc. SODA 2020}, pages 1010--1018, 2020.

\bibitem{Lovasz1972}
L.~Lov\'asz.
\newblock Normal hypergraphs and the perfect graph conjecture.
\newblock {\em Discrete Math.}, 2(3):253--267, 1972.

\bibitem{MnichT2018}
M.~Mnich and E.~Teutrine.
\newblock Improved bounds for minimal feedback vertex sets in tournaments.
\newblock {\em J. Graph Theory}, 88(3):482--506, 2018.

\bibitem{MnichVWV2016}
Matthias Mnich, Virginia Vassilevska~Williams, and L\'{a}szl\'{o}~A. V\'{e}gh.
\newblock A $7/3$-approximation for feedback vertex sets in tournaments.
\newblock In {\em Proc. ESA 2016}, volume~57 of {\em Leibniz Int. Proc.
  Informatics}, pages 67:1--67:14, 2016.

\bibitem{Speckenmeyer1990}
Ewald Speckenmeyer.
\newblock On feedback problems in digraphs.
\newblock In {\em Proc. WG 1990}, volume 411 of {\em Lecture Notes Comput.
  Sci.}, pages 218--231, 1990.

\bibitem{Vegh}
L{\'a}szlo V{\'e}gh.
\newblock Personal communication, as stated in [2].

\bibitem{YouWC2017}
Jie You, Jianxin Wang, and Yixin Cao.
\newblock Approximate association via dissociation.
\newblock {\em Discrete Appl. Math.}, 219:202--209, 2017.

\end{thebibliography}
\end{document}